\documentclass[a4paper]{article}

\usepackage{localeng}
\usepackage{hyperref}
\let\eps=\varepsilon
\DeclareMathOperator{\KS}{\mathrm{C}}
\DeclareMathOperator{\KP}{\mathrm{K}}
\DeclareMathOperator{\KA}{\mathrm{KA}}
\DeclareMathOperator{\Dim}{\mathrm{Dim}}
\newcommand{\cnd}{\mskip 0.8mu | \mskip 0.8mu}
\newtheorem*{theorem}{Theorem}
\newtheorem*{lemma}{Lemma}
\theoremstyle{remark}
	\newtheorem{remark}{Remark}

\title{Conditional normality and finite-state dimensions revisited}
\author{Alexander Shen\thanks{LIRMM, Univ Montpellier, CNRS, Montpellier, France\\ \texttt{sasha.shen@gmail.com}, \texttt{alexander.shen@lirmm.fr}.\\ Supported by ANR grant ANR-21-CE48-0023 FLITTLA.}}
\date{}

\begin{document}
\maketitle
\begin{flushright}
\emph{To Jarkko Kari, on the occasion of his 60th birthday}
\end{flushright}

\begin{abstract}
The notion of a normal bit sequence was introduced by Borel in 1909; it was the first definition of an individual random object. Normality is a weak notion of randomness requiring only that all $2^n$ factors (substrings) of arbitrary length~$n$ appear with the same limit frequency $2^{-n}$. Later many stronger definitions of randomness were introduced, and in this context normality found its place as ``randomness against a finite-memory adversary''. A quantitative measure of finite-state compressibility was also  introduced (the \emph{finite-state dimension}) and normality means that the finite state dimension is maximal (equals~$1$).

Recently Nandakumar, Pulari and S (2023) introduced the notion of \emph{relative} finite-state dimension for a binary sequence with respect to some other binary sequence (treated as an \emph{oracle}), and the corresponding notion of \emph{conditional} (relative) \emph{normality}. (Different notions of conditional randomness were considered before, but not for the finite memory case.) They establish equivalence between the block frequency and the gambling approaches to conditional normality and finite-state dimensions.

In this note we revisit their definitions and explain how this equivalence can be obtained easily by generalizing known characterizations of (unconditional) normality and dimension in terms of compressibility (finite-state complexity), superadditive complexity measures and gambling (finite-state gales), thus also answering some questions left open in the above-mentioned paper.

\end{abstract}

\section{Introduction}

Let us start by recalling what (unconditional) Borel normality is and how it is related to finite-state gambling.  (For the details, proofs, and references to original papers see, e.g., \cite{ks}.)

\subsection{Block frequencies}

Let $\alpha=a_1a_2a_3\ldots$ be a binary sequence. Fix some $k$ and split this sequence into blocks of length $k$. For every $N$ consider the first $N$ blocks ($kN$ bits in total) and, for each $k$-bit string, consider its frequency among these blocks. If all the frequencies converge to $1/2^k$ as $N\to \infty$, and this happens for every $k$ (block size), the sequence $\alpha$ is called \emph{normal}. 

More general, we consider the entropy $H_{k,N}(\alpha)$ of the distribution on the first $N$ blocks of length $k$. Then we consider $\liminf$ and $\limsup$ of this entropy as $N$ goes to infinity. For normal sequences the distribution converges to the uniform distribution, so both $\liminf$ and $\limsup$ are $k$ (for every $k$). This gives a characterization of normal sequences.  In a general case we divide $\limsup_N$ and $\liminf_N$ of $H_{k,N}(\alpha)$ by $k$ and take the limit as $k\to\infty$ (one can prove that this limit always exists and coincides with the infimum over $k$). These two limits  (for $\liminf$ and $\limsup$) are called respectively \emph{finite-state dimension} of $\alpha$, denoted by $\dim_{\mathrm{FS}}(\alpha)$, and \emph{strong finite-state dimension} of $\alpha$, denoted by $\Dim_{\mathrm{FS}}(\alpha)$. The letters ``FS'' stand for ``finite-state''; we omit them in the sequel since we do not consider other dimensions. 

As we have noted, we have $\dim(\alpha)=1$ (and therefore $\Dim(\alpha)=1$) for a normal sequence $\alpha$. One can prove that the reverse is also true: if $\dim(\alpha)=1$, then $\alpha$ is normal ($\Dim(\alpha)=1$ is not enough).

Instead of splitting $\alpha$ into $k$-bit blocks, we may consider a sliding window of size $k$, and consider the distribution of $k$-bit strings among the first $N$ non-aligned blocks  (thus using only $N+k-1$ bits of the sequence instead of $Nk$). Then we proceed as before; one can prove that the finite-state dimension and strong finite-state dimension remain the same (as for splitting into aligned blocks).

\subsection{Gambling and gales}\label{subsec:gambling-gales}

The finite-state dimensions of a bit sequence $\alpha$ have natural interpretation in terms of gambling against $\alpha$. Consider a guessing game where a gambler starts with initial capital $1$ and, before getting the next bit of $\alpha$, makes a bet by splitting her capital into bets on $0$ and $1$; one part is lost and the other is doubled. In general, the gambler strategy can be fully described by a function $m$ defined on binary strings: $m(X)$ is the capital of the gambler after having played against binary string $X$. The gambling rules mean that this function is a \emph{martingale}, i.e., $m(\varepsilon)=1$ for the empty string $\varepsilon$, and $m(X)=[m(X0)+m(X1)]/2$ for every binary string $X$ and its two possible extensions $X0$ and $X1$. We restrict our attention to \emph{finite-state} martingales that correspond to gambling strategies with \emph{finite memory}. This means that at every moment gambler is in one of finitely many states; the state determines the ratio of the bets on $0$ and~$1$ (we assume that this ratio is a rational number); the next state is determined by the previous one and the observed bit.

Intuitively, if a sequence $\alpha$ is ``non-random'', the regularities in $\alpha$ can be used to gamble against it; more non-randomness means faster growth of the martingale along the sequence. In the ultimate case when the sequence is all zeros and we bet on zero all the time, the capital after $n$ bets is $2^n$. On the other hand, a cautious gambler that always splits the capital evenly between $0$ and $1$ keeps her original capital untouched (the martingale is always $1$). We measure the growth in logarithmic scale and consider the ratio $(\log m(a_1\ldots a_n))/n$. In the two extreme cases mentioned this ratio is $1$ and $0$ respectively. In the general case, for a given sequence $\alpha=a_1a_2\ldots$ and a finite-state gambler $G$ (with corresponding martingale $m_G$), we consider
\[
\limsup_n \frac{\log m_G(a_1\ldots a_n)}{n}
\quad \text{and}\quad
\liminf_n \frac{\log m_G(a_1\ldots a_n)}{n}
\]
Then we take supremum over finite-state gamblers $G$ (we are interested in the \emph{best} gamblers, not in the worst ones) and get two quantities that are directly related to finite-state dimension: one can prove that the first one (with $\limsup$) is $1-\dim(\alpha)$, and the second one (with $\liminf$) is $1-\Dim(\alpha)$.

Usually this procedure is explained in a slightly different language: instead of measuring the capital growth rate, we impose a ``tax'' by multiplying the capital by $2^{s-1}$ after each game (a zero tax if $s=1$, and $50\%$ tax for $s=0$; in the latter case we need to bet all the capital on the correct bit all the time just to keep the capital unchanged). In other words, we consider finite-state \emph{$s$-gales} that satisfy the equation
\[
2^{s-1} m(X)=\frac{m(X0)+m(X1)}{2},
\quad \text{or} \quad
2^s m(X)= m(X0)+m(X1).
\]
Then we consider the infimum of the values of $s$ that allow some finite-state gambler to \emph{win} despite this tax. For strong finite-state dimension \emph{winning} means that the value of $s$-gale \emph{tends to infinity}, for finite-state dimension winning means that it is \emph{unbounded}.

One can consider also combinations of finitely many martingales ($l$-account gales). This means that the initial capital is split between $l$ finite-state gamblers and each of them plays independently using her own capital. Note that while the total memory of $l$ gamblers is finite, they cannot be always replaced by one finite-state gambler, since the combined ratio depends on the current capitals of $l$ individual gamblers and does not necessarily belong to any finite set. Still this generalization does not affect the finite-state dimension and strong finite-state dimension. (For finite-state dimension it is easier to see than for strong finite-state dimension, but it is true in both cases.)

\subsection{Introducing oracles \textup(conditions\textup)}

Informally speaking, conditional randomness of a bit sequence $\alpha=a_1a_2\ldots$ with respect to some other bit sequence $\beta=b_1b_2\ldots$ (considered as a condition) means that we cannot find any regularities in $\alpha$, or cannot win the gambling game against $\alpha$ \emph{even if we are given access to $\beta$ as an oracle} (for free). For most notions of algorithmic randomness the conditional versions are well understood (and follow the standard relativization scheme used in computability theory: all computable objects are replaced by $\beta$-computable objects, where $\beta$ is used as an oracle). However, for finite-state randomness (normality) this general scheme does not work: it is not clear \emph{a priori} how the finite memory gambler should access the oracle. If we think about this more closely, we see that there are some important choices --- at least two of them --- that should be made when we define conditional normality.

\subsubsection{Look-ahead while gambling}

First, we need to decide whether we allow looking ahead or not. Recall that $\alpha$ is normal if a finite memory gambler cannot win against $\alpha$.  Now we need to allow the gambler to access some oracle $\beta$. Taking into account the general spirit of the definition of normality, it is natural to assume that this access is somehow local and synchronized with the betting position in the sequence~$\alpha$. 

But do we permit looking ahead or not? Assume that $\alpha=a_1a_2\ldots$ is some ``perfectly random'' sequence and $\beta = b_1b_2\ldots$ is $\alpha$ shifted to the right: $a_{n}=b_{n+1}$; does this make $\alpha$ non-normal with respect to $\beta$? The answer suggested by Nandakumar, Pulari and S in~\cite{nps} is \emph{yes}, we do allow looking ahead for $O(1)$ future oracle bits, and $\alpha$ in this example is not $\beta$-normal. More precisely, for each finite-state gambler some look-ahead constant $c$ is fixed; the gambler sees next $c$ bits $b_{i+1}\ldots b_{i+c}$ of the oracle  $\beta$ (as well as $b_i$ itself) when gambling on some bit $a_i$ of~$\alpha$.

\subsubsection{Negligible frequencies}

How can we define conditional normality of some sequence $\alpha$ in terms of block frequencies?  The condition is some other sequence $\beta$. It is natural to cut both sequences $\alpha$ and $\beta$ in the same places, so (for a given block size $k$) we get pairs of $k$-bit strings $(A,B)$. (To get $N$ those pairs we use $Nk$-bit prefixes of $\alpha$ and $\beta$.) 

What do we require from this sequence of pairs?  The second components (being part of an oracle that we do not control) can be arbitrary. What we could require is that the conditional distribution on first components, for every fixed value of the second component, is uniform.  Note that this conditional distribution may be undefined if some string $B$ never appears as a condition (as the second component of a pair), or can be rather irrelevant if $B$ appears only finitely many times as a second component. 

But what if some $k$-bit string $B$ appears as a condition infinitely many times, but still is extremely rare? Should we require something about the conditional frequencies of the first components for the second component~$B$? The answer suggested in~\cite{nps} is \emph{no}, and this is logical if we are interested only in the base of the exponent that measures the capital growth: very slow growth is then indistinguishable from no growth. More formally, the following requirement for conditional normality is introduced in~\cite{nps}:  for every $N$ we take $N$ first pairs and get a distribution on $\mathbb{B}^k\times\mathbb{B}^k$; we require that the conditional entropy of the first component with the second component as a condition converges to $k$ as $N\to\infty$. In this way we combine the requirements for all conditions using frequencies of different conditions as their weights (so extremely rare conditions do not matter).

\subsection{What follows}

It turns out that the two approaches to relative normality described above give the same notion --- this is one of the main results of~\cite{nps}. (Note that the block approach is consistent with the permission to look ahead: condition block is provided in its entirety, not bit by bit.) Moreover, it is shown there that these two approaches lead to equivalent definitions of conditional finite-state dimension and conditional strong finite-state dimension. The proof in~\cite{nps} is quite technical, and in this note we show how these results can be obtained in a natural way by extending other definitions of finite-state dimensions (with finite-state complexity, superadditive functions etc.) to the conditional case and by proving equivalence of all these extended definitions (thus anwering also some questions left open in~\cite{nps}).

As we mentioned, a detailed discussion of different definitions of normality and finite-state dimension (for non-conditional case) can be found in~\cite{ks}; basic information about Kolmogorov complexity (that is sometimes used in our arguments) can be found, say, in~\cite{s}. Often the changes needed for the conditional case are minimal; still we try to provide a self-contained account of what happens for conditional case, explaining all essential steps.

In the next section we formulate several equivalent definitions of conditional finite-state dimensions. Then in Section~\ref{sec:proof} we explain why they are indeed equivalent to each other.

\section{Formal definitions}\label{sec:def}

Let $\alpha=a_1a_2\ldots$ and $\beta=b_1b_2\ldots$ be two bit sequences. We provide several (equivalent) definitions of conditional finite-state dimension $\dim(\alpha\cnd \beta)$ and conditional strong finite-state dimension $\Dim(\alpha\cnd\beta)$.

\subsection{Block frequencies}

Fix some positive integers $k$ (block size) and $N$ (the number of blocks). Then split first $Nk$ bits of $\alpha$ and $\beta$ into $N$ strings of length $k$. Denote these strings by $A_1,\ldots, A_N$ and $B_1,\ldots,B_N$, and consider $N$ pairs $(A_1,B_1),\ldots,(A_N,B_N)$. For a uniformly random $i\in \{1,\ldots,N\}$ we get a pair of jointly distributed random variables $\mathcal{A},\mathcal{B}$ whose values are $k$-bit strings $A_i$ and $B_i$. Let $H_{k,N} (\alpha\cnd \beta)$ be the conditional entropy of the variable $\mathcal{A}$ with $\mathcal{B}$ as the condition. Then let
\[
\dim (\alpha\cnd \beta)=\lim_k \liminf_N \frac{H_{k,N}(\alpha\cnd\beta)}{k},
\quad
\Dim (\alpha\cnd \beta)=\lim_k \limsup_N \frac{H_{k,N}(\alpha\cnd\beta)}{k}
\]
As we will see, the $\lim_k$ is guaranteed to exist and can be replaced by $\inf_k$ without changing the dimensions.

An alternative version (that gives the same dimensions) uses non-aligned blocks (sliding window) instead of aligned ones: we define $A_i$ and $B_i$ as
\[
A_i=a_i\ldots a_{i+k-1} 
\quad \text{and} \quad
B_i=b_i\ldots b_{i+k-1}
\]
and then continue as before.

\subsection{Finite-state gamblers}

Consider a game with infinite number of rounds. Initially the gambler has capital~$1$. In $i$th round, before seeing $a_i$, the gambler sees $b_i,\ldots,b_{i+c}$ (where $c$ is a constant that is chosen by the gambler before starting the game), and splits her current capital into two parts labeled $0$ and $1$ (the ratio is a rational number chosen by the gambler). Then $a_i$ is shown to the gambler, the part labeled $a_i$ is doubled and the other part is lost. 

In general, the gambler strategy $G$ is determined by the constant $c$ and a function that gets the history of the game (i.e., bits $a_1,\ldots,a_{i-1}$, $b_1,\ldots,b_{i+c}$) and determines the  ratio of the bets. However, we consider only \emph{finite-state} strategies. Finite-state strategy is defined in a natural way: there is some finite set $S$ of \emph{states} including some \emph{initial} state, and there is a transition function that for every current state $s$, last observed bit $a_{i-1}$ and oracle bits $b_i\ldots b_{i+c}$ determines the ratio of bets (a rational number) and the next state.

For every gambling strategy $G$ (finite-state or not) we consider the corresponding \emph{martingale} function whose value $m_G(a_1\ldots a_N\cnd b_1\ldots b_{N+c})$ is the gambler's capital after playing against $a_1\ldots a_N$ with oracle bits $b_1\ldots b_{N+c}$ according to the rules. Then we define dimensions as 
\[
\dim(\alpha\cnd\beta) = 1 - \sup_G \limsup_N \frac{\log m_G(a_1\ldots a_N\cnd b_1\ldots b_{N+c})}{N}
\]
and 
\[
\Dim(\alpha\cnd\beta) = 1 - \sup_G \liminf_N \frac{\log m_G(a_1\ldots a_N\cnd b_1\ldots b_{N+c})}{N}
\]
where the supremum is taken over all finite-state strategies $G$ (with arbitrary constants $c$).

\subsection{Automatic complexity}

Conditional Kolmogorov complexity $\KS_D (A\cnd B)$ for bit strings $A$ and $B$ is defined as
\[
\KS_D(A\cnd B) = \min\{ |P|\colon (A,B,P) \in D\}.
\]
Here $|P|$ stands for the length of the binary string $P$ and $D$ is a ternary relation on binary strings $A,B,P$ called a \emph{description mode}; we read $(A,B,P)\in D$ as ``$P$ is a description of $A$ given $B$''. In Kolmogorov complexity theory we require that $D$ is a function, namely, the first argument $A$ is a function of two other arguments $B$ and $P$, and this function should be computable; in other words, $D$ should be (computably) enumerable. For every $D$ that has these properties we consider the corresponding function $\KS_D$. There exist optimal description modes~$D$  that make $\KS_D$ minimal up to $O(1)$ additive term; we fix one of them and call $\KS_D(A\cnd B)$ the (plain) \emph{Kolmogorov complexity of $A$ given $B$}.

Now we want to adapt this definition to the finite-state case by defining a special class of ternary relations called \emph{automatic description modes} (here ``automatic'' is a synonym for ``finite-state'').

Let $G$ be some (finite directed) graph whose edges are labeled by pairs $(l,p)$ where $p$ is either a bit ($0/1$) or a special symbol $\eps$, and $l$ is either a pair of bits $(a,b)$ or a special symbol $\eps$. (The choice whether to use $\eps$ in a label or not is made separately for $p$ and for $(a,b)$ but not separately for $a$ and $b$.) For every path (walk) in $G$ we combine all labels on the edges and get three strings $A,B,P$ (made of corresponding bits; symbols $\eps$ are skipped). Note that $A$ and $B$ (but not $P$) in a triple corresponding to some path have the same length. 

In this way we get (for every labeled graph of the described type) a ternary relation $D(A,B,P)$ on binary strings. This relation is called an \emph{automatic description mode} if $A$ is a $O(1)$-valued function of $B$ and $P$, i.e., if there is some constant $c$ such that for every $B$ and $P$ there are at most $c$ values of $A$ such that $(A,B,P)\in D$. Then we define $\KS_D(A\cnd B)$ \emph{for binary strings $A$ and $B$ of the same length} as before, i.e., as the minimal length of $P$ such that $D(A,B,P)$ holds.

Then dimensions can be defined (for sequences $\alpha=a_1a_2\ldots$ and $\beta=b_1b_2\ldots$) as
\[
\dim(\alpha\cnd\beta) = \inf_D \liminf_N \frac{\KS_D(a_1\ldots a_N\cnd b_1\ldots b_N)}{N}
\]
and
\[
\Dim(\alpha\cnd\beta) = \inf_D \limsup_N \frac{\KS_D(a_1\ldots a_N\cnd b_1\ldots b_N)}{N}.
\]
The infimum is taken over all automatic description modes $D$.

\begin{remark} 
This definition may look strange at first: why we require the relation to be a \emph{multi-valued} function? This is done for an important technical reason: we do not fix the initial state (vertex) in our graph, and a path generating $(A,B,P)$ may have arbitrary first vertex (and arbitrary last vertex, too). In this setting the functionality requirement would be too strong. (The absence of the initial state is important for superadditivity, as we will see below.) On the other hand, the original definition of Kolmogorov complexity remains the same (up to $O(1)$-additive terms) if we use description modes that are $O(1)$-valued (computably) enumerable ternary relations on strings.  
\end{remark}

\begin{remark}
The connection between different notions of compressibility and entropy goes back to Shannon, and the relation between finite-state compression and block entropy was analyzed long ago in~\cite{zl} where the notion of finite-state (strong) dimension appeared (under the name of \emph{compressiblity} and denoted by~$\rho$). See~\cite[Section 6]{ks} for the historic account of these and subsequent developments. What seems to be new in our approach (used in~\cite{ks} for non-conditional complexity and adapted here for conditional complexity) is that we consider only decompression (as it is done in algorithmic information theory) and choose the technical details in a special way to guarantee superadditivity that plays an important technical role in the proofs.
\end{remark}

\subsection{Superadditive complexity measures}\label{subsec:super}

This definition of dimension is more technical, and its motivation comes from the proof of the equivalence between the entropy and complexity definitions of dimension. Still we provide it here since it is quite simple and instructive.

Let $K(A\cnd B)$ be a function on pairs of binary strings of the same length ($|A|=|B|$) with non-negative real values. We call it a \emph{superadditive complexity measure} if it satisfies two conditions:
\begin{itemize}
\item
$K(A_1A_2\cnd B_1B_2)\ge K(A_1\cnd B_1)+K(A_2\cnd B_2)$ for every strings $A_1,A_2,B_1,B_2$ such that $|A_1|=|B_1|$ and $|A_2|=|B_2|$ (\emph{superadditivity});
\item
there exists some constant $c$ such that for every number $m$ and string $B$ the number of strings $A$ such that $K(A\cnd B)\le m$ does not exceed $c2^m$ (\emph{calibration}).
\end{itemize}
\begin{remark}
The first condition somehow reflects the finite-state requirement; it says, roughly speaking,  that no information transfer happens between two stages where $A_1$ and $A_2$ are described. (It is technically important that we do not allow any increase and have no $O(1)$ term in the inequality). As we will see, for every automatic description mode the function $\KS_D$ is superadditive.  Note also that for the standard Kolmogorov complexity this requirement does \emph{not} hold: $K(A_1A_2\cnd B_1B_2)$ could be much smaller than $K(A_1\cnd B_1)+K(A_2\cnd B_2)$ if $A_1$ and $A_2$ share a lot of information. (The formula for the complexity of pairs guarantees \emph{sub}additivity with logarithmic precision.)
\end{remark}
\begin{remark}
The second requirement says that $K$ should not be ``too small'' (e.g., the zero function is not a superadditive complexity measure even though it is superadditive). It is fulfilled for conditional Kolmogorov complexity and for automatic complexity since (for a given condition) each string $P$ may be description of only one (for Kolmogorov complexity) or $O(1)$ (for automatic complexity) strings.
\end{remark}

Now the finite-state dimensions are defined as follows: 
\[
\dim(\alpha\cnd\beta) = \inf_K \liminf_N \frac{K(a_1\ldots a_N\cnd b_1\ldots b_N)}{N}
\]
and
\[
\Dim(\alpha\cnd\beta) = \inf_K \limsup_N \frac{K(a_1\ldots a_N\cnd b_1\ldots b_N)}{N}.
\]
The infimum is taken over all superadditive complexity measures $K$.

\subsection{Finite-state a priori complexity}

The next definition works as a bridge between the complexity-based definition and the gambling definition, but it has an independent motivation, too. In algorithmic information theory, in addition to complexity, we also consider the \emph{a priori probability}, a maximal semicomputable distribution. It exists in two versions: (1)~on integers, the so-called \emph{discrete} a priori probability, and (2)~on binary sequences, the \emph{continuous} a priori probability. Our finite-state version is similar to the continuous case, so let us recall the definition for this case.

Consider a probabilistic machine (= randomized algorithm) $M$ without input that probilistically generates output bits sequentially (with arbitrary delays). Then for every binary string $X$ we consider the probability $m_M(X)$ of the event ``at some moment the string $X$ is the output of the machine'' (after that moment some other output bits may or may not appear). By definition $m_M(\eps)=1$ for empty string $\eps$ (since at the beginning there was no output), and $m_M(X)\ge m_M(X0) + m_M(X1)$ since the two events in the right hand side are disjoint subsets of the event in the left hand side. The function $m_M$ is lower semicomputable (i.e., can be computably approximated from below),  and  every lower semicomputable function that has the above-mentioned properties is equal to $m_M$ for some machine $M$. There exist a ``universal'' machine $M$ that makes $m_M$ maximal (by simulating every other machine with some positive probability); we fix such a machine $M$ and call the function $m_M$ the \emph{continuous a priori probability}, sometimes denoted by $\mathsf{a}(X)$. Its minus logarithm $\KA(X)= - \log \mathsf{a}(X)$ is called the \emph{a priori complexity} of $X$ and coincides with plain Kolmogorov complexity $\KS(X)$ with logarithmic precision, up to $O(\log|X|)$ terms. (See~\cite{suv} for details.)  This definition does not mention oracles, but for the general (not finite-state) case relativization is easy: oracle access is allowed for all algorithms mentioned in the definition (in our case, for machine $M$ and for algorithms that approximate the function $m_M$ from below).

Now we modify this definition to adapt it to the finite-state case taking into account the synchronous access to bits of condition/oracle. Consider a graph with finitely many vertices (states).  For every state, and for every condition bit ($0$ or $1$), there are two outgoing edges that correspond to output bits $0$ and $1$, labeled with rational non-negative numbers with sum $1$. In total for every vertex we have two pairs with sum $1$: one for condition bit $0$ and one for condition bit $1$, each of these four numbers is written on the edge going to the next state after the transition. These numbers are interpreted as probabilities to emit $0$ and $1$ (and change the state according to their edges) being in a given state and reading a given condition (input) bit.

Let $M$ be such a labeled graph. Assume that some state $s$ and some input string $B=b_1\ldots b_t$ are fixed. Then a probabilistic process is defined: we start at state $s$ and use each input bit to determine the probabilities of transitions emitting output bits and changing the state. Then we consider the probability $m_{M,s}(A\cnd B)$ that an output string is~$A$.  We get (for each $B$) a probability distribution on strings $A$ that have the same length as the string $B$. 

Then, as in the algorithmic information theory, we define the \emph{finite-state a priori complexity}:
\[
\KA_M(A\cnd B) = -\log_2 \max_s m_{M,s}(A\cnd B);
\]
as we will see, the maximum over $s$ is  technically important since it makes the function $\KA_M$ superadditive.
Now the dimensions can be defined as follows:
\[
\dim(\alpha\cnd\beta) = \inf_{M,c} \liminf_N \frac{\KA_M(a_1\ldots a_N\cnd b_{1+c}\ldots b_{N+c})}{N}
\]
and
\[
\Dim(\alpha\cnd\beta) = \inf_{M,c} \limsup_N \frac{\KA_M(a_1\ldots a_N\cnd b_{1+c}\ldots b_{N+c})}{N}.
\]
Note that this definition, like the gambling one (and unlike the others) mentions the look-ahead constant explicitly.

\begin{remark}
In this definition we do not allow the random process to arrive to different states but output the same bit (for a given input bit). This is natural if we think about connection with finite-state gamblers (the probabilities to emit $0$ and $1$ correspond to the parts of the gambler's capital bet on $0$ and $1$), but a more general definition can also be used (see the remark at the end of Section~\ref{subsec:apriori}).
\end{remark}

\subsection{Main result}

\begin{theorem}[Nandakumar, Pulari, S~\cite{nps}, extended version]
All the five definitions of conditional finite-state dimension and conditional strong finite-state dimension given above are equivalent.
\end{theorem}

In the next section we prove this equivalence (and also mention some other variations that still give equivalent definitions), starting with block frequencies and automatic complexity.

\section{Proofs}\label{sec:proof}

\subsection {From block frequencies to automatic complexity}\label{subsec:block2auto}

We start by proving that the block entropy dimensions cannot be smaller than the corresponding automatic complexity dimensions. In other words, we assume that for some $k$ the limit ($\liminf$ or $\limsup$) of conditional $k$-bit block entropies is smaller than some threshold $\tau$, and we construct an automatic description mode for which the corresponding limit is also smaller than $\tau$.

The basic tool here is Shannon--Fano code. For a random variable with $n$ values that have probabilities $p_1,\ldots,p_n$ there exists a prefix-free code with codeword lengths $m_1,\ldots,m_n$ where $m_i = \lceil - \log p_i\rceil$, and therefore the average code length $\sum p_i m_i$ does not exceed $\sum p_i (-\log p_i) +1$, i.e., $H+1$ where $H$ is the entropy of the distribution. In our setting we apply Shannon--Fano code to the distribution on $k$-bit blocks (so $n=2^k$).

The (natural) decoder for a prefix-free code has  finite memory. The prefix-free codewords cover some subtree of a finite binary tree. We start at the root and follow the directions determined by input bits (and output nothing) until we reach the codeword. Then we go through a chain of states (no input, only output) and after producing $k$ output bits (the encoded block) return to the root and are ready to decode the next block. Decoding would be unique if we fix the initial state and final vertices of the path to be the root. In our setting, when they are not fixed, we have O(1) outputs per input, since some $O(1)$-length prefix of the sequence (that we want to decode) brings the automaton to the initial state, and there are only $O(1)$ possibilities for output when the codeword is read in the input sequence. The Shannon--Fano theorem guarantees that this automaton provides a description whose length (per block) is close to the entropy of the block distribution. But there are still three problems.

\begin{itemize}
\item
The prefix code provided by Shannon--Fano theorem depends on the distribution. But we have to deal with all distributions on blocks (we fixed $k$, the block size, but the distribution on first $N$ blocks may depend on $N$ and be arbitrary). 

\item We need to take into account the conditions. For every condition block we have some conditional distribution that is different for different conditions, and the corresponding prefix-free code also depends on the condition. We can access the condition bits in the automaton, but do we get the required $O(1)$-bound in the definition of automatic description mode?

\item Finally, the Shannon--Fano theorem has some overhead: the average length of code (per block) is bounded by $H+1$, not $H$.  This overhead ($+1$) then will be divided by block size $k$, but even $1/k$ overhead is bad for us: we need exact equality.
\end{itemize}

How do we deal with these problems?\footnote{In~\cite{ks}, unfortunately, this is done incorrectly for the case of strong dimensions (sorry!): Lemma 19.1 is not enough since it provides prefix code for $N$-bit strings that depends on $N$, and this cannot be used to construct one finite automaton for all $N$. We need to be more careful and construct a finite family of codes for strings of fixed length, as it is done in Lemma~\ref{universal} and Lemma~\ref{universal-quasi}.} For the first one, we use not one prefix-free code but a family of prefix-free codes provided by the following lemma.

\begin{lemma}\label{universal}
For arbitrary finite alphabet~$X$ there exists a finite family of binary prefix-free codes for $X$ such that for every distribution $P$ on $X$ some prefix code from the family has average code length at most $H(P)+1$.
\end{lemma}

Here the average (in the average code length) is taken over $P$. The lemma says that the good code can be found in some finite family (whose size depends only on the alphabet size).

\begin{proof}
The simplest way to prove this lemma is to note that Huffman's construction of an optimal code gives a code where the maximal length of a codeword is bounded by the alphabet size (the reduction step decreases the alphabet size by $1$,  and adding trailing $0/1$ increases the length by $1$). So there are only finitely many optimal Huffman codes for an alphabet of a given size (whatever the distribution~is).

One can also use compactness argument: optimal code for some distribution is close to optimal for some neighborhood of that distribution, and these neighborhoods cover the compact space of all distributions (a simplex), so there is a finite cover. (We need to take into account the distributions on a simplex boundary whose Shannon--Fano code does not have a codeword for some letter; they should be replaced by other codes  to cover a neighborhood, and this causes $O(1)$ overhead, so we have a weaker bound with some other constant instead of $1$, but this does not matter.) 
\end{proof}

Using the lemma, we can now solve the first problem mentioned above by taking a disjoint sum of decoding automata for all codes in the family (this is possible, since we allow $O(1)$-valued functions anyway). Moreover, since we need to consider conditional codes, we construct a decoding automaton \emph{for each function that maps conditions to codes from our finite family}; the number of these automata increases exponentially, but it is still finite (for a given block size $k$), so their disjoint sum is still a finite automaton. Each decoding automaton (for every mapping) is constructed as follow: it \emph{guesses} the condition (before reading the first bits of the description) and then decodes the description according to the prefix-free code for the guessed condition. At the same time the automaton checks whether the guessed condition matches the actual one (there is no edge to the next state if there is a discrepancy). In this way all the wrong guesses do not add anything to the description mode relation, so the $O(1)$-requirement is satisfied.

The last problem is to deal with $O(1)$ overhead (unavoidable in the Shannon--Fano theorem). The solution is to double the block size until the overhead is negligible. Let us see how the entropies for $2k$-bit blocks are related to entropies for $k$-bit blocks. Consider $N$ first blocks of length $2k$ (covering $2kN$ bits both in the sequence and in the condition), and let $\mathcal{A}$ and $\mathcal{B}$ be the corresponding random variables whose values are $2k$-bit sequences. We can split each variable in two halves: $\mathcal{A}=\mathcal{A}^1 \mathcal{A}^2$, $\mathcal{B}=\mathcal{B}^1 \mathcal{B}^2$.  Each of the variables $\mathcal{A}^1, \mathcal{A}^2, \mathcal{B}^1,\mathcal{B}^2$ is defined on the same probability space (with $N$ elements) and takes $k$-bit values. Then
\[
H(\mathcal{A}^1\mathcal{A}^2\cnd \mathcal{B}^1\mathcal{B}^2)\le H(\mathcal{A}^1\cnd \mathcal{B}^1\mathcal{B}^2)+H(\mathcal{A}^2\cnd \mathcal{B}^1\mathcal{B}^2)\le H(\mathcal{A}^1\cnd \mathcal{B}^1)+H(\mathcal{A}^2\cnd \mathcal{B}^2).
\]
We will show now that the sum in the right hand side is smaller than $2\tau$ (assuming the entropy for $k$-bit blocks is smaller than $\tau$). Consider a random variable $z$ which takes values $1$ and $2$ with equal probability. After $z$ is chosen, take blocks $\mathcal{A}^z$ and $\mathcal{B}^z$ randomly from the corresponding distribution (for pairs $\mathcal{A}^1$, $\mathcal{B}^1$ or for pairs $\mathcal{A}^2$, $\mathcal{B}^2$). We get three jointly distributed variables $z$, $\mathcal{A}^z$, $\mathcal{B}^z$. If we omit $z$ and consider the other two variables $\mathcal{A}^z$ and $\mathcal{B}^z$, we get a distribution on pairs of $k$-bit blocks that is exactly the distribution for $2N$ pairs that we obtain when using blocks of size $k$, and $H(\mathcal{A}^z \cnd \mathcal{B}^z)$ is smaller than $\tau$. On the other hand, the smaller quantity $H(\mathcal{A}^z\cnd \mathcal{B}^z,z)$ equals 
\[
H(\mathcal{A}^1\cnd \mathcal{B}^1)\Pr[z=1] + H(\mathcal{A}^2\cnd \mathcal{B}^2)\Pr[z=2] = \frac{1}{2}\left(H(\mathcal{A}^1\cnd \mathcal{B}^1) + H(\mathcal{A}^2\cnd \mathcal{B}^2)\right)
\]
so $H(A^1\cnd B^1) + H(A^2\cnd B^2)<2\tau$, as we promised. So we may switch to larger blocks without increase in entropy (per bit), and the $O(1)$ overhead becomes twice less important. This construction then can be repeated (or we may switch directly from $k$ to any multiple of $k$).

This almost finishes the proof; there is only one subtle point that we missed. The doubling argument assumes that the value of $N$ where the block entropy was small, is \emph{even} (then we may compare the $k$-block entropy with $2k$-block entropy). However, we are not guaranteed that the values of $N$ that are important for $\liminf$, are even. To overcome this difficulty, we note that entropy function (for a given block size) is uniformly continuous on a simplex where it is defined, and adding/deleting one block for large $N$ causes only small change in the frequencies, so we may add or delete one block and still have entropy smaller than $\tau$. (In fact, for $\limsup$ we do not have this problem, since in this case entropy is smaller than $\tau$ for all large $N$, including large even values of~$N$.) 

This argument proves inequalities between automatic and block entropy definitions of dimension and strong dimension. Indeed, we have seen that for every $\eps$ and for every block size $k$ there exists an automatic description mode $M$ such that
\[
\liminf_t \frac{\KS_M(a_1\ldots a_t\cnd b_1\ldots b_t)}{t} \le \frac{\liminf_N H_{k,N}(\alpha,\beta)}{k} +\eps,
\]
When (for fixed $k$) we take $\inf_M$, as required by the complexity definition of dimension, $\eps$ disappears. Then we may add $\inf_k$ in the right hand side:
\[
\inf_M\  \liminf_t \frac{\KS_M(a_1\ldots a_t\cnd b_1\ldots b_t)}{t} \le \inf_k\frac{\liminf_N H_{k,N}(\alpha,\beta)}{k},
\]
and the same is true for $\limsup$ instead of $\liminf$ (in both places).

\subsection{From automatic complexity to block frequencies}

Now we want to prove that dimensions defined in terms of block frequencies cannot be larger than corresponding dimensions defined in terms of automatic complexity. Informally speaking, if block entropy for large enough blocks is big, then the automatic complexity should be also big. In this proof our main tool is superadditivity (defined in Section~\ref{subsec:super}).

\begin{lemma}
For every automatic description mode $M$ the corresponding function $\KS_M$ is superadditive: $\KS_M(A_1A_2\cnd B_1B_2)\ge \KS_M(A_1\cnd B_1)+\KS_M(A_2\cnd B_2)$ for every strings $A_1,A_2,B_1,B_2$ such that $|A_1|=|B_1|$ and $|A_2|=|B_2|$.
\end{lemma}

\begin{proof}
Assume that $A_1A_2$ with condition $B_1B_2$ has some description $P$ with respect to $M$. This means that we can read $A_1A_2$ having condition $B_1B_2$ along some path in $M$. Choose a point in this path when all bits of $A_1$ are produced and all bits of $B_1$ are read. (Recall that input and output bits are synchronized.) This point divides $P$ into $P_1P_2$, where $P_1$ is a description of $A_1$ with condition $B_1$, and $P_2$ is a description of $A_2$ with condition $B_2$, so we get the required bound for the sum of complexities.
\end{proof}

\begin{remark}
This proof shows why we need to consider arbitrary initial (and final) state in the definition of automatic complexity: if we required all paths to start from some fixed initial state, then $P_2$ would not be a description. 
\end{remark}

Superadditivity gives us a natural way to get a lower bound for automatic complexity. Let $M$ be some automatic (conditional) description mode. To get a lower bound for $\KS_M(a_1\ldots a_t \cnd b_1\ldots b_t)$ for some prefixes of sequences $\alpha$ and $\beta$, we may take some block size $k$, split the prefixes into $k$-bit blocks and use superadditivity. Then we get a lower bound, namely, the sum of conditional automatic complexities $\KS_M(A_i\cnd B_i)$ for the blocks $A_i$ and $B_i$ (with the same $i$). Here $k$ may be arbitrary; it is possible that $t$ is not a multiple of $k$, then the last incomplete block can be discarded (its automatic complexity is non-negative), and we use only $u=\lfloor t/k\rfloor$ complete blocks for the lower bound. 

Now it is convenient to use the language of Kolmogorov complexity. We know that $\KS_M(A\cnd B) \ge \KS(A\cnd B)-c$ for some $c$ (depending only on $M$) and for all $A$ and $B$ of the same length ($k$ in our case). Note that $c$ does not depend on $k$. Therefore, we get the lower bound $\sum_{i=1}^u\KS(A_i\cnd B_i) - cu$. It is convenient to switch to prefix complexity $\KP(A_i\cnd B_i)$, it can be larger by $O(\log k)$, so we get the lower bound 
\[
\KS_M(a_1\ldots a_t \cnd b_1\ldots b_t) \ge \sum_{i=1}^u \KP(A_i\cnd B_i)-uO(\log k).
\]

Recall that the prefix complexity provides a prefix-free encoding for all strings (and therefore for $k$-bit blocks, too), and conditional prefix complexity (that we have here) provides a family of prefix-free encodings (for each condition we have some encoding). Now, for each value of condition (i.e., for each $k$-bit block $B$) we apply Shannon lower bound for average code length for the blocks $A_i$ such that $B_i=B$. Then these bounds are combined for all conditions, and the weights are frequencies of the conditions, so 
\[
\sum_{i=1}^u \KP(A_i\cnd B_i) \ge uH(\mathcal{A}\cnd \mathcal{B}),
\]
where $\mathcal{A}$ and $\mathcal{B}$ are random variables $A_i$ and $B_i$ for uniformly distributed $i\in \{1,\ldots,u\}$, i.e., the variables considered in the definition of block entropy, where $H(\mathcal{A}\cnd \mathcal{B})$ was denoted by $H_{k,u}(\alpha\cnd\beta)$. Combining all these inequalities and dividing by the number of bits $t$ (we assume for now that $t$ is a multiple of $k$) we get 
\[
\frac{\KS_M(a_1\ldots a_t \cnd b_1\ldots b_t)}{t}\ge \frac{H_{k, t/k}(\alpha\cnd\beta)}{k}- O\left(\frac{\log k}{k}\right)
\]
As $t$ goes to infinity (and $k$ is fixed), the incomplete block influence becomes negligible, and for $\liminf$ and $\limsup$ it does not matter, so we forget about our assumption ($t$ is a multiple of $k$) and have
\[
\liminf_t \frac{\KS_M(a_1\ldots a_t \cnd b_1\ldots b_t)}{t}\ge \liminf_N\frac{H_{k,N}(\alpha\cnd\beta)}{k}- O\left(\frac{\log k}{k}\right)
\]
for every block size $k$, and similar inequality with $\limsup$ in both sides. Then we take $\limsup_k$ in the right hand side, and the term $O(\log k / k)$ disappears. (Note that the left hand side does not depend on $k$.) After that the right hand side does not depend on $M$ (recall that the error term for fixed $k$ did depend on $M$), and we conclude that
\[
\inf_M\  \liminf_t \frac{\KS_M(a_1\ldots a_t \cnd b_1\ldots b_t)}{t}\ge \limsup_k \ \liminf_N\frac{H_{k,N}(\alpha\cnd\beta)}{k},
\]
and this is the required inequality for two definitions of finite-state conditional dimension. The same argument works for $\limsup$.

\begin{remark}
Note that we have $\limsup_k$ in the right hand side while we had $\inf_k$ in the reverse inequality of  Section~\ref{subsec:block2auto}. This shows that the limit over $k$ always exists and coincides with infimum over $k$ (so we may use $\lim_k$, $\limsup_k$, $\liminf_k$, or $\inf_k$ in the definition of dimension in terms of block entropy).
\end{remark}

\begin{remark}\label{avoiding-complexity}
One could wish to avoid using notions from Kolmogorov complexity theory in this argument. For that we could note that if every string can be a description for at most $O(1)$ strings, we can add $O(1)$ bits  to encoding to get a unique code (but not a prefix-free one).  To get a prefix-free code from the unique code we prepend each encoding by the prefix-free encoding of its length. This gives prefix-free code with only logarithmic overhead. (Note that we may assume without loss of generality that the description length is bounded by $ k+O(1)$, so the overhead is $O(\log k)$ for $k$-bit blocks.)
\end{remark}

\subsection{Non-aligned blocks}

Let us show that we may as well use non-aligned blocks in the definitions of finite-state dimensions. Fix some block length $k$. There are $k$ ways to split our sequences into $k$-bit blocks if we start with some incomplete block; they correspond to $k$ possible boundaries positions modulo $k$. For each position of boundaries we have some distribution on the first $N$ pairs of blocks, and these distributions are quite unrelated to each other.  The distribution of non-aligned blocks is the mixture of these~$k$ distributions (with equal weights).

If, for some $N$, \emph{one of these $k$ distributions} has small (conditional) entropy, then this distribution can be used to construct an automaton that gives small automatic complexity. We need only to change the automatic description mode by adding a possibility to emit at most $k$ arbitrary bits reading arbitrary condition bits at the beginning (and not using any description bits); this does not destroy $O(1)$-value property of the description relation. 

On the other hand, the superadditivity argument can be used for \emph{each of $k$ possible splittings}. This implies that we can use both minimal  or maximal entropy (among all $k$ ways of splitting) in the block entropy definition of dimensions. Note also that mixing these $k$ distributions for $k$-bit blocks we get the distribution for non-aligned $k$-bit blocks.

Let us compare the \emph{entropy of this average} (mixed) distribution and the \emph{average of entropies} of these $k$ distributions. The latter average is the entropy of the average distribution \emph{with additional condition} that is the choice of the splitting positions modulo $k$. This additional condition is a uniformly distributed variable with $k$ values, so its entropy is $\log k$ and  therefore adding this condition can decrease the entropy at most by $\log k$. It remains to note that $\log k /k \to 0$ (as $k\to \infty$) to see that we may use non-aligned distributions  and get equivalent entropy definitions for finite-state dimension and strong dimension.

\subsection{Superadditivity criterion}

Another byproduct of the argument above is the characterization of finite-state dimensions in terms of superadditive complexity measures. Indeed, the automatic complexity function $\KS_M$ for each automatic description mode $M$ is superadditive and calibrated. On the other hand, for every superadditive calibrated function we can apply the same argument that worked for the automatic complexity. The only technical problem is that we cannot claim that superadditive calibrated function is an upper bound for Kolmogorov (conditional) complexity. Still it is easy to see that it is an upper bound (with $O(1)$-precision) for Kolmogorov complexity with some oracle (and the rest of the argument remains unchanged). For example, we may use the function $K$ (an arbitrary superadditive complexity measure given to us; nothing is assumed about its computability) as an oracle, then for every $m$ and for every condition we can enumerate all strings that have the $K$-value at most $m$ for that condition; we have $O(2^m)$ of them and they can be encoded by programs of length $m+O(1)$ with oracle $K$. (Alternatively, we may use the combinatorial argument sketched in the Remark~\ref{avoiding-complexity} above to avoid Kolmogorov complexity completely.)

The superadditivity criterion will be useful also for the analysis of two remaining definitions that use finite-state a priori complexity and gambling.

\subsection{Finite-state a priori complexity}\label{subsec:apriori}

Let us show that a priori complexity definition of dimensions gives the same dimensions as the definitions we already studied. One inequality is guaranteed by the following lemma. Recall that finite-state a priori complexity was defined in terms of labeled graphs of a special type.

\begin{lemma}
For every labeled graph $M$ the function $\KA_M(A\cnd B)$ is superadditive.
\end{lemma}

\begin{proof}
To prove this lemma, we need to provide an upper bound for the value $m_{M,s}(A_1A_2\cnd B_1B_2)$ for arbitrary state strings $A_1$, $B_1$ (of equal length), for arbitrary strings $A_2,B_2$ (that also have equal length), and for arbitrary initial state~$s$. To generate $A_1A_2$ with condition $B_1B_2$, starting from some state $s$, the probabilistic process should first generate $A_1$ with condition $B_1$ (coming to some random state $s'$), and then generate $A_2$ with condition $B_2$ starting from $s'$. Therefore, 
\[
m_{M,s}(A_1A_2\cnd B_1B_2) \le m_{M,s}(A_1\cnd B_1)\cdot
\max_{s'} m_{M,s'}(A_2\cnd B_2);
\]
now, taking maximum over $s$ and then taking logarithms, we get the desired inequality.~%
\end{proof}

The superadditivity property allows us to get a lower bound for 
\[
\KA_M(a_1\ldots a_N\cnd b_{1+c}\ldots b_{N+c})
\]
in the same way as we did for automatic complexity, i.e., by combining the bounds for $k$-bit blocks. Like automatic complexity, $\KA_M$ is an upper bound for Kolmogorov complexity, but for a different version of it (the a priori complexity $\KA$); this version also differs from prefix complexity by at most $O(\log k)$ for $k$-bit blocks, so this is not a problem. The problem is that in the lower bound we have complexities where the condition is shifted by $c$ positions, i.e., 
\[
\KA(a_{i+1}\ldots a_{i+k}\cnd b_{i+1+c}\ldots b_{i+k+c})
\]
and not
\[
\KA(a_{i+1}\ldots a_{i+k}\cnd b_{i+1}\ldots b_{i+k}),
\]
while we know how to get a lower bound in terms of block entropy only for the latter (non-shifted) version. However, the condition in the first case has only $c$ bits that are missing in the condition in the second case, so the first complexity may be smaller at most by $O(c)$. Since $c$ is fixed, and the block length $k$ goes to infinity, this change is negligible, and we conclude that dimensions defined in terms of finite state a priori complexity are not smaller that dimensions defined in terms of block entropy.

To get an inequality in the other direction, let us consider first the unconditional case. Consider a finite bit sequence split into several $k$-bit blocks; assume we have $N$ blocks $A_1,\ldots,A_N$, each of length $k$. We want to have an upper bound for a priori complexity $\KA_M(A_1\ldots A_N)$ of this sequence for some finite-state probabilistic process~$M$.  We construct $M$ starting from some probability distribution $P$ on blocks; it will generate blocks independently according to $P$. The probabilistic process $M$ starts from the root of a full binary tree of height $k$ and then traverses the tree generating a $k$-bit block with the probability taken from~$P$. (For example, the root has two outgoing edges, generating bits $0$ and $1$, and the probability to generate $0/1$ is the total $P$-probability of all blocks that start with $0/1$.) After $k$ bits (forming a block) are generated, the process returns to root, and the next block is generated in the same way independently from the previous ones. 

For such a process, the probability to generate a sequence $A_1\ldots A_N$ is a product of the probabilities of individual blocks according to the chosen distribution, so 
\[
\KA_M(A_1\ldots A_N) \le N \sum_B Q(B)\log\frac{1}{P(B)} 
\]
where the sum is taken over all $k$-bit blocks, $Q(B)$ is the frequency of $B$ in $A_1,\ldots,A_N$, and $P(B)$ is the probability of block $B$ in the distribution used to construct $M$.

As we know from the proof of the Shannon--Fano theorem (using Jensen's inequality, the convexity of the logarithm function), the right hand side, for given $Q$, is minimal when $P$ equals $Q$, and is $H(Q)$. But we need to construct one process $M$ that gives good bounds for all the prefixes of $\alpha$. So we are in the same situation as in Section~\ref{subsec:block2auto} and use a similar lemma (for distributions instead of codes).

\begin{lemma}\label{universal-quasi}
For arbitrary finite alphabet $X$ there exists a finite family $\mathcal{P}$ of probability distributions on $X$ with rational values such that for every distribution $Q$ on $X$ there exist some distribution $P$  from the family $\mathcal{P}$ such that
\[
\sum_{x\in X} Q(x)\log\frac{1}{P(x)} \le H(Q) + 1.
\]
\end{lemma}

\begin{proof}
This lemma is a consequence of the similar lemma in Section~\ref{subsec:block2auto}, since we may consider family of codes from that lemma and convert them to distributions (codelength $l$ corresponds to probability $2^{-l}$; Kraft's lemma guarantees that the sum of these probabilities does not exceed $1$ and we may increase them to get exactly the sum~$1$).
\end{proof}

\begin{remark}
Recalling the Shannon--Fano theorem and its proof, one can think about probability distributions as ``relaxed codes'': when the code length $l_i = \log (1/p_i)$ is not required to be integer; so our task for $\KA_M$ is easier that the corresponding task for $\KS_M$.
\end{remark}

\begin{remark}
Due to this relaxation, we can replace the constant $1$ in the statement of the lemma by arbitrarily small positive number~$\eps$. This strong version of the lemma can be proven using the compactness argument: for each $Q$ we may find some $P$ (not on the boundary of the simplex of distributions) that $\eps$-serves this $Q$; it has to $\eps$-serve $Q$ together with some small neighborhood, and then we can find a finite cover using compactness. So the increase in block size (that was necessary to avoid overhead for codes) may be avoided now. But we already know that it is possible, so there is no reason to avoid it.
\end{remark}

Now, using this lemma, we finish the proof in the same way as in Section~\ref{subsec:block2auto}. Having a finite family of graphs corresponding to the distributions from the Lemma, we take a disjoint union of them. Our definition (where we take maximum probability over all initial states) guarantees that all the bounds provided by the individual processes remain valid for the disjoint union.

This was the argument for unconditional case, but we have to deal with the conditions. Our random process uses the probability distribution on blocks that depends on the condition block.  So our random process uses the look-ahead bits (we need $k$ of them, where $k$ is the block size) and a finite memory that stores them.  More precisely, we consider all mappings of $\mathbb{B}^k$ to $\mathcal{P}$ where $\mathbb{B}^k$ is the set of all $k$-bit strings, and $\mathcal{P}$ is the family of distributions from the Lemma. For every mapping $F$ of this type and for all possible values of the first condition block we construct a random process with look-ahead $k$ that for every condition $B$ generates bits according to $F(B)$. (Note that the first condition block is known, the second condition block is read in parallel with generating the first $k$ bits, and so on.)  The we take the disjoint union of those processes for all $F$ and for all values of the first condition block. 

\begin{remark}
One can use a bigger class of random processes in the definition of a priori finite-state complexity. Namely, one may allow several outgoing edges with the same output letter and the same input letter but going to different vertices. (In this way the current state is not determined by input and output bits.) This wider class still gives a superadditive function $\KA_M$, so the first part of the proof remains valid (and the second part becomes easier when we extend the class).
\end{remark}

\subsection{Gambling definition}

The gambling characterization is essentially a reformulation of the a priori complexity characterization. In general, there is one-to-one correspondence between martingales and measures. Recall that martingale satisfies the condition
\[
m(X) = \frac{m(X0)+m(X1)}{2},
\]
while measures satisfy similar condition without factor $1/2$. Every martingale therefore can be represented as
\[
m(X) = \frac{P(X)}{2^{-|X|}},
\]
where $P$ is some measure on the Cantor space of infinite binary sequences, $P(X)$ is the probability for the sequence to start with $X$, and $|X|$ stands for the length of $|X|$. In other words, a martingale is the ratio of two measures: $P$ and the uniform measure on the Cantor space.

The same correspondence can be explained as follows. Assume we have some measure $P$ on the Cantor space. The corresponding gambling strategy divides the initial capital in proportion $P(0)\,{:}\,P(1)$ between bets on $0$ and $1$.  The same is done at the following stage; we use the conditional probabilities of $0$ and $1$ after already known prefix to determine the next bets. Finite-state martingales correspond to output measures of finite-state random processes of the type we considered, \emph{with fixed initial state}. This correspondence works as well for the conditional case.

It remains to show that the argument in the previous section (its second part, where we construct the upper bound for the finite-state a priori complexity) can be modified to get a random process with fixed initial state. Recall that we had a disjoint sum of many random processes, one for each function $F$ that maps condition blocks to distributions on blocks. So some modification is indeed needed, without it we have not one gambler but a finite family of gamblers, each with some capital that can be used for bets. It is quite possible that for prefixes (of $\alpha$ and $\beta$) of different length different gamblers become rich, and in the gambling characterization we need \emph{one gambler for all prefixes}. 

Note that a sum of martingales is a martingale, so we can consider the set of gamblers as one financial institution that makes bets. The problem is that this new combined martingale is no more a finite-state one, since the resulting proportion for the entire institution depends not only on the proportions chosen by individual gamblers, but also on their current capitals, so infinitely many ratios are possible. So to make the argument above valid for gambling, one could extend the notion of finite-state gambling by considering several gamblers with separated accounts. As we have noted in Section~\ref{subsec:gambling-gales}, this notion is known under the name ``$k$-account $s$-gale''. 

However, we may modify the construction to get one finite-state gambler instead of many. Fix some (large) number $T$. Let us agree that after each $T$ games the individual gamblers redistribute their money evenly between them and continue to play with the average. Now the entire institution will behave as a finite-state gambler with a rather large (but finite) number of states. Of course, this redistribution could be a loss for an individual gambler: in the worst case her capital is divided by some constant (the number of individual gamblers) once in $T$ games. But still $T$ can be arbitrarily large, so the change in the base of the growth exponent can be arbitrarily small, and this finishes the proof.

\section*{Final remarks}

\textbf{Questions}. Returning to the original motivation, one could ask the following natural questions about other possible approaches:

\begin{itemize}

\item Is there some reasonable notion of on-line conditional normality where the look-ahead access to condition bits is not allowed?

\item Is there some reasonable notion of ``strict'' normality when even a slow increase of capital is not allowed?

\item We considered normality with respect to the uniform Bernoulli measure. Is is possible to extend this notion and its equivalent characterizations to some wider class of measures (Markov chains, or non-uniform  Bernoulli measures, or some classes of product measures)?

\end{itemize}

\textbf{Acknowledgements}.  I am grateful to all the colleagues in LIRMM (Montpellier) and other places with whom I discussed the notions of normality and finite-state dimensions, especially to Ruslan Ishkuvatov and Alexander Kozachinskiy. I am grateful to the anonymous referee (\emph{Theoretical Computer Science}) for important remarks.

\end{document}